\documentclass[aps,amsfonts,twoside,amssymb,superscriptaddress,twocolumn,pra]{revtex4-1}

\usepackage[latin1]{inputenc}
\usepackage[T1]{fontenc}
\usepackage{amsmath,amssymb,amstext,amsfonts,amsxtra}
\usepackage{amsthm}
\usepackage{graphicx}
\usepackage{mathrsfs}
\usepackage{enumerate}
\usepackage{bbm}
\usepackage{times}
\usepackage{hyperref}

\newcommand{\tr}{\operatorname{tr}}
\def\idty{{\leavevmode\rm 1\mkern -5.4mu I}} %  unit operator

\def\ket #1{\vert #1\rangle}

\def\ketbra #1#2{\vert #1\rangle \langle #2\vert}
\def\kettbra#1{\ketbra{#1}{#1}}
\def\tr{\mathop{\rm tr}\nolimits}

\newcommand*{\cH}{\mathcal{H}}

\newcommand*{\cO}{\mathcal{O}}
\newcommand*{\cP}{\mathcal{P}}
\newcommand*{\cS}{\mathcal{S}}

\newcommand*{\cT}{\mathcal{T}}

\newcommand*{\cX}{\mathcal{X}}

\newcommand\pr[1]{\ensuremath{\mathrm{Pr}[#1]}}
\newtheorem{lem}{Lemma}
\newtheorem{thm}{Theorem}

\newcommand{\pass}{\textnormal{pass}}
\newcommand{\fail}{\textnormal{fail}}
\newcommand{\tot}{\textnormal{tot}}
\newcommand{\PE}{\textnormal{PE}}

\newcommand{\se}{\textnormal{sec}}
\newcommand{\IR}{\textnormal{IR}}
\newcommand{\Key}{\textnormal{key}}

\newcommand{\ex}{\textnormal{ex}}
\newcommand{\sq}{\textnormal{sq}}
\newcommand{\asq}{\textnormal{asq}}
\newcommand{\vac}{\textnormal{vac}}
\newcommand{\dd}{\textnormal{d}}
\newcommand{\co}{\textnormal{c}}
\newcommand{\diag}{\textnormal{diag}}

\def\cH{{\mathcal H}}
\def\cS{{\mathcal S}}

\newcommand{\ppass}{p_{\textnormal{pass}}}

\usepackage{color}
\definecolor{myred}{rgb}{1,0,0}

\definecolor{myblue}{rgb}{0,0,0.8}

\definecolor{myyellow}{rgb}{0.9,0.8,0}

\definecolor{mygreen}{rgb}{0,0.6,0}

\definecolor{myorange}{rgb}{0.6,0.6,0}

\definecolor{mycerul}{rgb}{0,0.6,1}

\begin{document}

\title{Reverse Reconciliation Continuous Variable Quantum Key Distribution Based on the Uncertainty Principle}

\author{Fabian Furrer}
\affiliation{Department of Physics, Graduate School of Science,
University of Tokyo, 7-3-1 Hongo, Bunkyo-ku, Tokyo, Japan, 113-0033.}
\email[]{furrer@eve.phys.s.u-tokyo.ac.jp}

\begin{abstract}
A big challenge in continuous variable quantum key distribution is to prove security against arbitrary coherent attacks including realistic assumptions such as finite-size effects. Recently, such a proof has been presented in [Phys. Rev. Lett. 109, 100502 (2012)] for a two-mode squeezed state protocol based on a novel uncertainty relation with quantum memories. But the transmission distances were fairly limited due to a direct reconciliation protocol. We prove here security against coherent attacks of a reverse reconciliation protocol under similar assumptions but allowing distances of over $16$~km for experimentally feasible parameters. We further clarify the limitations when using the uncertainty relation with quantum memories in security proofs of continuous variable quantum key distribution. 
\end{abstract}

%\pacs{PACS number: 03.67.Dd, 03.67.Hk}

\maketitle

\section{Introduction}
The most advanced quantum information technology is quantum key distribution (QKD), which is the art of using quantum properties to distribute a secure key between two remote parties. Its challenge lies in the combination of state of the art experimental implementations and newly developed quantum information theoretic principles to ensure its security. There exist two different implementations both of which have different benefits. More established is the encoding of the information in a quantum system with discrete degrees of freedoms, as, e.g., the polarization of a photon. Such discrete variable protocols are usually based on single photon sources and detectors with the latter suffer from low efficiency at room temperature and being susceptible to loopholes (see, e.g.,~\cite{Makarov11,Eisaman2011}). The advantage of such protocols is that conditioned on the arrival of a single photon, the channel noise is generally weaker allowing for long distances. 

An alternative implementation encodes the information into the quadratures of the electromagnetic field (see the recent review~\cite{Weedbrook2011} and references therein). Since the quadratures have a continuous spectrum they are called continuous variable QKD protocols. Compared to discrete variable protocols, they are based on variants of homodyne detection which is a robust and efficient measurement technique already used in current telecommunication systems. Although CV QKD systems are secure against blinding attacks, they are particularly vulnerable to manipulations of the phase reference signal (local oscillater) (see, e.g., ~\cite{Haeseler2008,jouguet2013b}). Since the information is directly encoded in the phase and amplitude of the laser beam, the fiber losses severely damp the transmitted signal and with that the encoded information.  Nevertheless, it was shown in~\cite{Grosshans03}  that a key can be generated for arbitrary losses using reverse reconciliation protocols. This has recently also been experimentally demonstrated against restricted attacks~\cite{jouguet2012}. 

Up to recently, the security of continuous variable QKD protocols has only been analyzed in the asymptotic limit assuming an infinite number of communication rounds (see, e.g.,~\cite{Grosshans03,Weedbrook04,Leverrier2011}). For protocols based on a Gaussian phase and amplitude modulation this simplifies the security analysis tremendously. For instance, so-called collective attacks in which each signal is attacked independently and identically are as powerful as general (coherent) attacks~\cite{Renner_Cirac_09}. Moreover, it has been shown that Gaussian collective attacks are optimal among all collective attacks~\cite{Cerf2006,Navascues2006}. But these powerful results can no longer be applied if finite-size effects due to only a finite number of communication rounds are considered. And furthermore, even under a restricted set of collective Gaussian attacks a significantly lower key rate is obtained for feasible block lengths~\cite{Leverrier2010}.

A big challenge in continuous variable QKD is to prove security against coherent attacks including all finite-size effects. Since the Hilbert space of the system is infinite-dimensional certain techniques that are standard for discrete variable security proofs cannot be applied. For instance, the exponential quantum de-Finetti theorem~\cite{renner_nature} or the post-selection technique~\cite{Renner_Postselection} that are used to lift security against collective attacks to security against coherent attacks do not directly apply in infinite dimensions (c.f.~\cite{Renner_Cirac_09}). Recently, the post-selection technique has been extended in order to apply it to continuous variable QKD~\cite{leverrier13}, but its practical implementation relies on a cumbersome symmetry step which is unpractical for real life applications. 

Another promising approach has been presented in~\cite{Furrer12} which is based on a newly extended uncertainty relation including the effect of entangled observers~\cite{berta09,Berta13}. The corresponding protocol is based on the distribution of entangled two-mode squeezed states and homodyne detection implemented in~\cite{Eberle2013}. The uncertainty relation allows to bound the information of an eavesdropper Eve solely by the correlation strength between the honest parties Alice and Bob. It has thus the advantage that no tomography, or equivalently, quantum channel estimations are necessary with the consequence of not relying on collective attacks. But in~\cite{Furrer12,Furrer12E} only losses up to $20$\% could be tolerated since a direct reconciliation protocol has been used. Moreover, the potential and limitations of the proof technique have not been fully investigated. 

Here, we show that using a reverse reconciliation protocol significantly higher losses of over $50$\% can be tolerated enabling transmission distances of over $16$ km including finite-size effects. This makes the protocol suitable for practical short distance continuous variable QKD providing security against coherent attacks. The security proof has the advantage that it does not require any assumptions on Alice's measurement device and is thus one-sided device independent.  

Compared to~\cite{Furrer12}, the reverse reconciliation protocol requires Bob to apply a test to control the energy of the incoming signal. The test is based on a beam splitter to reflect a negligible part of the signal which is then measured with a heterodyne detector. We then show that conditioned that the outcomes of the heterodyne detector are sufficiently small the probability of Eve using a large energy attack can be neglected. This test further allows one to overcome the problem that homodyne detectors only operate faithfully in a limited detection range. Moreover, we provide a new statistical estimation procedure that enables us to deal with high energy signals which was not possible in~\cite{Furrer12}.

We also clarify the theoretical limitations of the proof technique based on the extended uncertainty relation. In particular, we provide the optimal key rate in the asymptotic limit of an infinite number of exchanged signals and without statistical uncertainty. Unfortunately, it turns out that even under these ideal conditions the tolerated losses are limited. An investigation of the asymptotic key rate for a broad range of continuous variable protocols based on the uncertainty relation has recently been given in~\cite{Walk14}.

The paper is organized as follows. We start in Section~\ref{sec:KeyRate} by introducing the security definitions and the classical part of the protocol. This enables us to give a general formula for the key rate presented in~\eqref{eq:KeyLength}. In Section~\ref{sec:Setup}, we discuss the experimental setup and how the raw key is formed. The different steps of the protocol are then listed in Section~\ref{sec:Protocol} together with the assumptions. The main result is Theorem~\ref{thm:KeyLength} which gives the explicit formula for the key length. In Section~\ref{sec:KeyRates}, we present plots of the key rates for experimentally feasible parameters. The security analysis is given in Section~\ref{sec:SecAnalysis}. The tightness of the security proof is analyzed in Section~\ref{sec:Tightness}. Eventually, we conclude our results in Section~\ref{sec:Conclusion}.

\section{Security of a QKD Protocol and Finite-Key Rate} \label{sec:KeyRate}

\subsection{Security Definitions} 
A generic QKD protocol consists of two phases. The first phase is given by the quantum part and includes the transmission and measurement of the quantum system. The second phase is purely classical and consists of the extraction of a secure key from the measured data by means of classical post-processing. In the following, we consider an entanglement based scenario in which the source is trusted and located in Alice's laboratory. She then sends one part of the quantum system through a quantum channel to Bob. It is always understood that Alice's and Bob's laboratory's are closed, that is, no unwanted information can leak to an eavesdropper. Once all quantum systems are distributed, Alice and Bob perform measurements to obtain the data from which the raw keys $X_A$ and $X_B$ are formed. At the same time a parameter estimation test is done which concludes whether one proceeds with the key extraction or one aborts the protocol. 
Since the key generation is a statistical process, one can assign a probability $p_\pass$ to the event that the parameter test is passed. 

Given that the parameter estimation test is passed, Alice and Bob proceed with the classical post processing to generate the final keys $S_A$ and $S_B$, respectively.
Here, $S_A$ and $S_B$ are classical random variables which might be correlated with a quantum system $E$ hold by an eavesdropper. We denote the associated classical-quantum state by $\rho_{S_AS_BE}$. The state $\rho_{S_AS_BE}$ can conveniently be written as a classical quantum state
\begin{equation}\label{eq:KeyState}
\rho_{S_AS_BE} = \sum_{s_A,s_B} p(s_A,s_B) \kettbra{s_A,s_B} \otimes \rho_E^{s_A,s_B} \, ,
\end{equation}
where the classical values for the keys $s_A$ and $s_B$ are associated with orthonormal states $\ket{s_A,s_B}$ in a Hilbert space. Here, $p(s_A,s_B) $ denotes the distribution of keys and $\rho_E^{s_A,s_B} $ the quantum state of the eavesdropper conditioned on $S_A=s_A$ and $S_B=s_B$.  

We characterize a quantum key distribution protocol by its correctness and secrecy. For that we use a notion of security which is composable and based on the approach developed in~\cite{renner:04n,BenOr05,Renner_ComposableSecurity}. A protocol is called $\epsilon_c$\emph{-correct }if the probability that $S_A$ is not equal to $S_B$ is smaller than $\epsilon_c$: 
\begin{equation}\label{eq:correctness}
\text{Pr}[S_A \neq S_B ] \leq \epsilon_c \, .
\end{equation} 
Roughly speaking, a protocol is secret if the key $S_B$ is almost uniformly distributed and completely uncorrelated to Eve's system $E$. The ideal state is thus given by $\text{u}_{S_B} \otimes \rho_E$, where $\text{u}_{S_B}$ denotes the uniform distribution over all keys and $\rho_E$ is the reduction of the state~\eqref{eq:KeyState} to system $E$. We then say that a protocol is $\epsilon_s$\emph{-secret} if 
\begin{equation}
 (1-p_\pass) \, \Vert \rho_{S_BE} - \text{u}_{S_B} \otimes \rho_E \Vert_1   \leq \epsilon_s\,  ,
\end{equation} 
where $\Vert \cdot \Vert_1$ denotes the trace norm and the infimum is taken over all possible states of Eve's system. Eventually, a protocol is called $\epsilon_\se$\emph{-secure} if it is $\epsilon_c$-correct and $\epsilon_s$-secret with $\epsilon_c+\epsilon_s\leq \epsilon_\se$. Note that the above security definition is composable in the sense that security is guaranteed if any part of the key is used for any other cryptographic protocol. This follows from the monotonicity of the trace distance. 
 
\subsection{Classical Post-Processing}\label{sec:ClPostPro}
As discussed in the previous section, the classical post-processing transforms the raw keys $X_A$ and $X_B$ into the final keys $S_A$ and $S_B$. In the first step of this post-processing an information reconciliation protocol is applied to diminish the discrepancy of $X_A$ and $X_B$. It was shown in~\cite{Grosshans03} that it is beneficial for continuous variable QKD protocols to use a reverse reconciliation scheme in which Alice corrects her raw key $X_A$ in order to match $X_B$. This is especially crucial for long distance QKD. Throughout this paper, we assume that a one-way reverse reconciliation protocol is used in which $\ell_{\IR}$ bits of information about $X_B$ is sent to Alice via an authenticated classical channel. Given this information and $X_A$, Alice outputs a guess $ X^{\co}_A$ of $X_B$. 

In order to ensure correctness~\eqref{eq:correctness} for the raw keys $X^{\co}_A$ and $X_B$ (and thus for the keys), Bob applies a random function of a family of two-universal hash functions~\cite{Carter79,Wegman81} onto an alphabet of size $ 1/\epsilon_c$ on $X_B$. He then sends Alice over an authenticated public channel a description of the applied function together with the obtained value. This leaks additional $\log1/\epsilon_c$ bits of information, where the logarithm is always taken to base $2$. Alice applies the function to her corrected raw key $ X^{\co}_A$ and checks if the obtained value matches with the one from Bob. If this is the case, they proceed with the protocol otherwise they abort~\footnote{Note that in a practical situation one does not need to abort the protocol and may only supply more information in the reconciliation protocol until the test is passed.}.  This then ensures that the generated key is $\epsilon_c$ correct. 

In a second step of the classical post-processing the raw key is hashed to a sufficiently small alphabet by means of a family of two-universal hash functions such that the key is $\epsilon_{\se}$-secure. Let us assume that the output of the hash functions is a bit string of length $\ell_{\epsilon_\se}$.
For finite-dimensional $E$ systems it has been shown in~\cite{RennerPhD,Tomamichel10} that the length of the bit string $\ell_{\epsilon_\se}$ can be expressed by the smooth min-entropy $H^\epsilon_{\min}(X_B|E)$ which is related to the maximal probability that Eve guesses $X_B$ correctly (see Section~\ref{sec:UR}). This result has been extended in~\cite{Furrer11} to the case where Eve's system E is modeled by an infinite-dimensional Hilbert space, which is necessary for applications to continuous variable systems. In particular, it holds that if $\epsilon \leq (\epsilon_s-\epsilon_1)/(2p_\pass)$, 
\begin{equation}\label{eq:KeyLength}
H_{\min}^\epsilon(X_B|E)_\rho-  \ell_{\IR} -  \log \frac{1}{\epsilon_1^2\epsilon_c}  +2 
\end{equation}
is a tight lower bound on the key length $\ell_{\epsilon_\se}$ with $\epsilon_{\se} = \epsilon_c + \epsilon_s$ (see, e.g.,~\cite{FurrerPhD} for details). The state $\rho_{X_B E}$ for which the smooth min-entropy is evaluated corresponds to the classical-quantum state describing the joint state of Bob's raw key $X_B$ and Eve's system $E$ conditioned that the protocol passes. 
The goal of the security analysis is to obtain a tight lower bound on~\eqref{eq:KeyLength} using the data collected in the parameter estimation step.

%%%%%%%%%%%%%%%%%%%%%%%%%%%%%%%%%%%%%%%%%%%%%%
%%%%%%%%%%%%%%%%%%%%%%%%%%%%%%%%%%%%%%%%%%%%%%
%%%%%%%%%%%%%%%%%%%%%%%%%%%%%%%%%%%%%%%%%%%%%%

\section{The Protocol and Key Rates}\label{sec:ProtocolKeyRates} 

%%%%%%%%%%%%%%%%%%%%%%%%%%%%%%%%%%%%%%%%%%%%%
%%%%%%%%%%%%%%%%%%%%%%%%%%%%%%%%%%%%%%%%%%%%%

\subsection{Experimental Setup and Generation of Data} \label{sec:Setup}

The protocol is similar to the one in~\cite{Furrer12} and consists of the distribution of an entangled two-mode squeezed state and homodyne detection first proposed in~\cite{Cerf01}. But additionally, Bob performs a test in order to estimate whether the incoming signal exceeds a certain energy threshold. This test allows one to exclude high energy eavesdropping attacks and to restrict onto a bounded measurement range. This is crucial in order to do finite statistics with reliable error bounds. The test requires only two additional homodyne detectors. 

The source is assumed to be in Alice's laboratory and generates a two-mode squeezed entangled state often referred to as an EPR state. This can be implemented by mixing two squeezed modes over a balanced beam splitter~\cite{Furusawa98}. The important characteristic of a two-mode squeezed state is that there are two quadratures with a phase difference of $\pi/2$ for which the two modes are highly correlated. We call these quadratures amplitude $Q$ and phase $P$ in the following. Alice then keeps one mode in her laboratory and performs at random an amplitude or phase measurement using a homodyne detector, where the probability for phase is $0<r<1$. The other mode is sent through a fiber to Bob who is as well performing randomly an amplitude or phase measurement with probability $1-r$ and $r$, respectively.  Due to the property of a two-mode entangled state, Alice's and Bob's measurement outcomes are highly correlated if they both perform amplitude or phase measurement and uncorrelated otherwise.  

Before Bob measures amplitude or phase of the incoming signal he performs an energy test. In particular, he mixes the signal with a vacuum mode $a$ using a beam splitter with almost perfect transmittance $T$. The reflected signal $a'$ is measured via heterodyne detection, that is, mode $a'$ is mixed with another vacuum mode $b$ by a balanced  beam splitter and homodyne detection is performed to measure amplitude of one output $q_{t^1}$ (mode $t^1$) and phase $p_{t^2}$ of the other output (mode $t^2$). The setup is illustrated in Figure~\ref{fig:Energy}. Bob then simply checks whether $|q_{t^1}|$ and $|p_{t^2}|$ is smaller than a prefixed value $\alpha$ for every incoming signal and aborts otherwise. In the following we denote the corresponding test by $\cT(\alpha,T)$. In Section~\ref{sec:Test}, we show that conditioned that $\cT(\alpha,T)$ passes the probability for large amplitude and phase measurements can be bounded. 

\begin{figure}\begin{center}\includegraphics*[width=8.8cm]{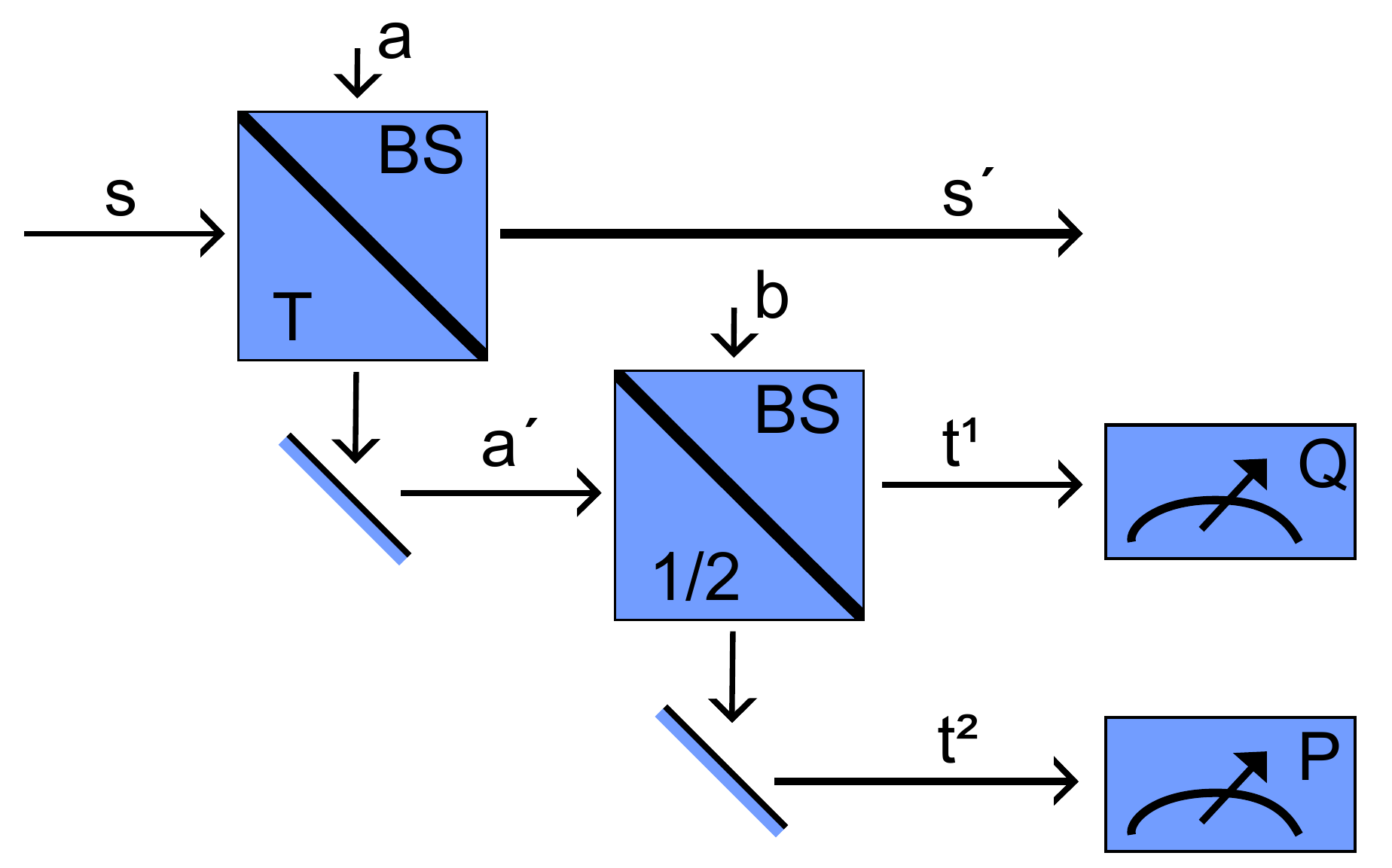}\caption{\label{fig:Energy} 
The diagram shows the measurement setup of Bob's test $\cT(\alpha,T)$. He mixes the incoming signal with a vacuum mode $a$ using a beam splitter with very low reflectivity $1-T$ and applies a heterodyne detection on the reflected signal $a'$. The test then consists of checking whether the absolute value of the outcome of the amplitude measurement of mode $t^1$ and the phase measurement on $t^2$ is smaller than $\alpha$.    
}\end{center}\end{figure}

While theoretically the spectrum of a homodyne measurement is the real line, any practical implementation is limited by a certain precision. We account for that by grouping outcomes into intervals of length $\delta$, where $\delta$ should be larger than the precision of the homodyne detector. We then choose an $M\geq 0$ smaller than the detector threshold and group the measurements into intervals
\begin{align*}
 I_1 & =(-\infty,-M+\delta] \, , \\
I_k &=(-M +(k-1)\delta,-M +k\delta]  , \  k=2,...,2M/\delta -1 \, , \\
I_{2M/\delta} & = (M-\delta,\infty) \, ,
\end{align*}
where we assume that $2M/\delta$ is in $\mathbb N$. We thus associate with any measurement result a value in $\cX=\{1,2,...,2M/\delta\}$. 

It is important for the protocol to have high correlations between Alice's and Bob's outcome in the index set $\cX$. But due to losses in the fiber, Bob's amplitude and phase quadratures $Q_B$ and $P_B$ will be damped. In order to account for that, we scale the quadrature measurements of Alice's detector $Q_A$ and $P_A$ before grouping them into the intervals using the transformations 
\begin{equation}
Q_A \mapsto \tilde Q_A = t_q Q_A \ \text{and} \  P_A\mapsto \tilde P_A = t_p P_A \, .
\end{equation}
The scaling factors $t_q$ and $t_p\leq 1$ are adjusted according to the channel losses in the transmission of the mode to Bob.

For the following, we also need a function to measure the strength of the correlations between two strings $X,Y\in \cX^N$. For that we introduce the average distance 
\begin{equation} \label{eq:Dist}
d(X,Y) = \frac 1 N \sum^N_{k=1} \vert X^k - Y^k\vert  \, . 
\end{equation}
We further define the average second moment of the difference between the strings by 
\begin{equation} \label{eq:DistVar}
d_2(X,Y) =  \frac 1 N \sum_{k=1}^N \vert X^k - Y^k\vert^2 \, .
\end{equation}
Moreover, we define average second moment for the discretized phase and amplitude measurements $X\in\cX^N$ by 
\begin{equation}\label{eq:Var}
\text{m}_2(X)=\frac 1N \sum_{k=1}^N (X^k-M/\delta)^2 \, .
\end{equation} 
Here, we subtract $M/\delta$ since in the absence of an eavesdropper the average value of $X$ will be (approximately) $M/\delta$ such that $\text{m}_2(X)$ simplifies to the variance. This holds because the first moments of the amplitude and phase measurements in the absence of Eve are $0$, which implies that the first moments of the discretized value will be approximately $M/\delta$.  
%

%%%%%%%%%%%%%%%%%%%%%%%%%%%%%%%%%%%%%%%%%%%%%%%%%%%%%%%%%
%%%%%%%%%%%%%%%%%%%%%%%%%%%%%%%%%%%%%%%%%%%%%%%%%%%%%%%%%

\subsection{The Protocol}\label{sec:Protocol}

The protocol depends on the total number of prepared two-mode squeezed states $N_\tot$, the probability that Alice and Bob perform a phase measurement $r$, the interval length for the data generation $\delta$, the threshold parameters $\alpha$ and $M$ (see Section~\ref{sec:Setup}), and a fixed value $d_0>0$ used in the parameter estimation test. All classical communication is assumed to be authenticated. The different steps in the protocol are as follows. 

\begin{enumerate}
	\item \textit{Distribution \& Measurement:} Alice prepares $N_\tot$ two-mode squeezed states and sends half of it to Bob upon which both measure for each mode phase with probability $r$ and amplitude with probability $(1-r)$. Moreover, Bob applies the test $\cT(\alpha,T)$, that is, he checks if $|q_{t^1}|,|p_{t^2}|\leq \alpha$ is satisfied for all of the $N_\tot$ incoming modes and aborts the protocol otherwise (see Figure~\ref{fig:Energy}). 
	
	\item \textit{Data Generation:} Alice and Bob publicly announce their basis choice. We count with $n$ and $k$ the number of events in which Alice and Bob both chose amplitude and phase measurement, respectively. From the measurement with the same basis choice, they use the amplitude and phase measurements to form $X_A$ and $X_B$ in $\cX^n$, and $Y_A^{\PE}$and $Y_B^{\PE}$ in $\cX^k$ according to Section~\ref{sec:Setup}. Alice and Bob further form a string containing all discretized phase measurements denoted by $Y_A^{P}$ and $Y_B^{P}$, respectively, where we assume that both have length $m$.

\item \textit{Parameter Estimation:} Using classical communication, they compute the distance $d^{\PE}=d(Y_A^{\PE},Y_B^{\PE})$ as in~\eqref{eq:Dist} and check if $d^{\PE} \leq d_0$. If this does not hold they abort the entire protocol. Otherwise, they proceed with the protocol and compute the second moment of the distance $V_d^\PE = d_2(Y_A^{\PE},Y_B^{\PE}) $ according to~\eqref{eq:DistVar}. Moreover, they individually compute the average second moments of all their phase measurements $V_{Y_A}^\PE=\text{m}_2(Y_A^P)$ and $V_{Y_B}^\PE=\text{m}_2(Y_B^P)$ according to~\eqref{eq:Var}.

\item \textit{Classical Post-Processing} They run a classical post-processing protocol as described in Section~\ref{sec:ClPostPro} by applying first a one-way reverse reconciliation protocol and secondly hash the corrected raw keys $X^{\co}_A$ and $X_B$ to final keys $S_A$ and $S_B$ of length $\ell$. 
\end{enumerate}

The crucial point is now to obtain a tight bound on the possible number of secure bits $\ell$ one can generate by the above protocol. Such a bound relies always on a set of assumptions. Such assumptions can, for instance, be a restriction on the attacks of the eavesdropper or simplifications used to model the experimental setup. We thus start, with a detailed description of our assumptions before presenting the key length formula. 

We always assume that Alice's and Bob's laboratory are secure and closed, that is, no unwanted information leaks from their laboratory. It is further very important to assume that all random numbers used for the basis choice and the classical post-processing are truly random and independent. This implies for instance that Alice's and Bob's basis choice are random and independent which is crucial for the security. 
While these assumptions are at the ground of most of the security analysis the following are specific for our measurement setup and security proof.   

\begin{enumerate} 
 \item[(A)] \textit{ Assumptions.} 
We assume that Bob's sequential measurement of the values in $\cX$ are independent and correspond to perfect amplitude and phase measurements of the intervals $I_k$ defined in Section~\ref{sec:Setup}. Hence, they can be modeled by integration of the spectrum of one-mode amplitude and phase operators with perfect phase difference of $\pi/2$~\footnote{The case of a small deviation from a phase difference of $\pi/2$ can easily be included.}. The same applies to Bob's test measurement performed in $\cT(\alpha,T)$.  
\end{enumerate}

We note that (A) includes the assumption that the local phase reference used by Bob is trusted. This can be practically justified by either monitoring the phase reference or generating it independently directly on Bob's side. For possible attacks on the local oscillator and countermeasures see, for instance,~\cite{jouguet2013b}. We emphasize that we do not make any assumptions on Eve's attacks and that there are no requirements on Alice's measurement device. The latter is sometimes referred to as one-sided device independent~\cite{tomamichel11}.

As we will discuss in details in Section~\ref{sec:SecAnalysis}, security will be inferred from the uncertainty principle with quantum memory for continuous variable systems~\cite{Berta13}. The principle says that Eve's information about the amplitude measurements is bounded by an overlap term of Bob's measurements expressed by
\begin{equation}\label{eq:overlap1}
c(\delta) \approx \delta ^2 /2\pi \, ,
\end{equation} 
and the uncertainty of Alice about Bob's phase measurement. The latter can be estimated using the distance $d_0$ and the function 
\begin{equation}\label{eq:gamma}
\gamma(t) = (t+\sqrt{1+t^2})\Big(\frac{t}{\sqrt{1+t^2}-1}\Big)^{t} \, .
\end{equation}

Moreover, we use the test $\cT(\alpha,T)$ to upper bound the probability that Bob measures an amplitude or phase quadrature larger than $M$ by  (see equation~\eqref{lem,eq:FailureProb}) 
\begin{equation}
  n \Gamma(M,T,\alpha) \propto  n \exp\big({{\ -\frac{(\mu M - \alpha)^2}{T(1+\lambda)/2}}}\big)\, 
\end{equation}
where $\mu=\sqrt{\frac{1-T}{2T}}$. Hence, the probability can be made sufficiently small by tuning the parameters $\alpha$, $T$, and $M$. Using large deviation bounds for the statistical estimation of the raw key sample we then obtain the following bound on the key length. 

\begin{thm}\label{thm:KeyLength}
Let us consider the above protocol with parameters $(N_\tot,r,\delta,M,\alpha,d_0)$ and assume that the conditions in (A) are satisfied. We further assume that the reconciliation protocol broadcasts $\ell_\IR$ bits of classical information and the correctness test is passed for two-universal hash functions onto an alphabet of size $1/\epsilon_c$. 
Then, if the protocol passes, an $\epsilon_{c}$-correct and $\epsilon_s$-secret key of length
\begin{equation}\label{thm,eq:KeyLengthCoherent}
 n [\log \frac{1}{c(\delta)}-\log \gamma(d_0 + \mu)] -  \ell_{\IR} -  \log \frac{1}{\epsilon_1^2\epsilon_c}  +2 ,
\end{equation}
can be extracted,
where 
\begin{equation}\label{eq:mu}
\mu= \sqrt{2\log \xi^{-1} } \frac{(n+k) \sigma_* }{k\sqrt{n}}       +  \frac{4 (M/\delta) \log \xi^{-1}}{3} \frac{n+k}{n k} \, ,
\end{equation}
with 
\begin{align}
\sigma_*^2 = & \ \frac kN (V_{d}^\PE - \frac kN (d^\PE)^2) \ + \frac kN \big( V_{Y_A}^\PE + V_{Y_B}^\PE + 2 \frac{\nu}{\delta^2} \big) \nonumber
 \\
& + 2 \frac kN \sqrt{(V_{Y_A}^\PE +\frac \nu{\delta^2}) (V_{Y_B}^\PE +\frac \nu{\delta^2})} \,  , \label{eq:Sigma}
\end{align}
for the smallest $\nu$ for which  
\begin{align}\label{eq:xi}
\xi =& \, \Big(\epsilon_s  - \epsilon_1 - 2\sqrt{2 n \ \Gamma(M,T,\alpha) }\Big)^2    \\
&  \, - 2\exp\Big(-2(\nu/M)^2\frac{n m^2}{(n+m) (m+1)} \Big)    \nonumber 
\end{align} 
is positive and $\epsilon_1 - 2\sqrt{1 -p_E^n } < \epsilon_s$. In the case that there is no $\nu$ such that $\xi$ is positive or $\epsilon_1 - 2\sqrt{2 \Gamma(M,T,\alpha)  } < \epsilon_s$ is not satisfied, the key length is $0$.  
\end{thm}
The proof of the above theorem will be given in Section~\ref{sec:SecAnalysis}. Before that we present some estimates of the obtained key rates for experimentally feasible parameters. 

\begin{figure}\begin{center}\includegraphics*[width=8.8cm]{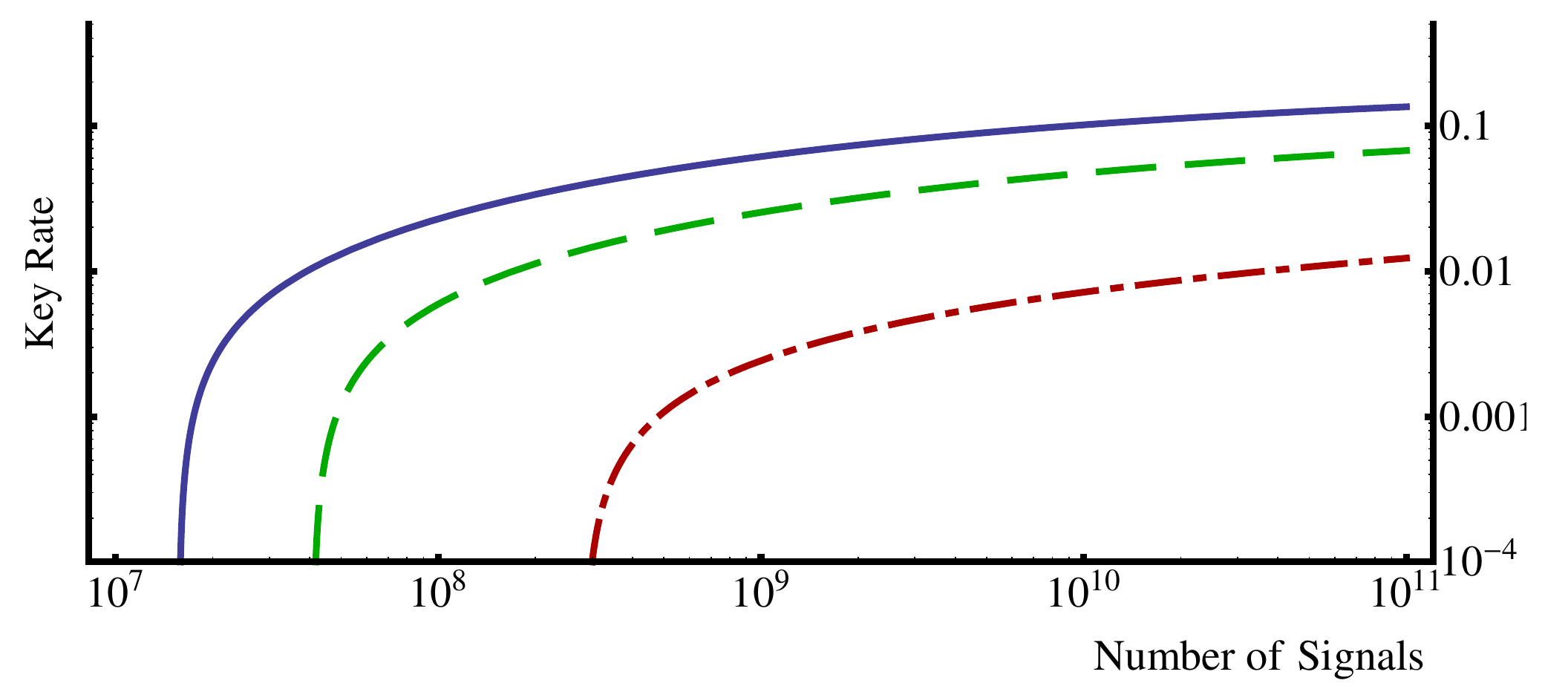}\caption{\label{fig:EC95} 
The plot shows the key rate $\ell/N_\tot$ for squeezing and antisqueezing of $11$dB and $16$dB and reconciliation efficency $\beta =0.95$ depending on the number of signals $N_\tot$. Bob's total losses $\eta_B$ are $0.45$ (solid line), $0.50$ (dashed line) and $0.55$ (dash-dotted line). Since the source is assumed to be in Alice's laboratory her losses are set to $\eta_A=0$. We set the excess noise $\eta_\ex = 0.01$, the security parameters to $\epsilon_s=\epsilon_c=10^{-9}$, and the test parameters to $T=0.99$ and $\alpha = 28$.    
}\end{center}\end{figure}

\subsection{Discussion of Key Rates} \label{sec:KeyRates}

For the following, we consider a two-mode squeezed state with squeezing $\lambda_\sq$ and antisqueezing $\lambda_\asq$ given by 
\begin{equation}\label{eq:CM}
\Gamma = \left( \begin{array}{cc}
\Gamma_A & \Gamma_{\text{cor}}  \\
 \Gamma_{\text{cor}}  & \Gamma_B 
 \end{array} \right) \,  ,
\end{equation}
where $\Gamma_A = \Gamma_B = a \idty$ and $ \Gamma_{\text{cor}} = \sqrt{a^2-b^2} Z$ with $a=\frac 12 (10^{\frac{\lambda_\sq}{10}} + 10^{\frac{\lambda_{\asq}}{10}})$, $b=10^{\frac{\lambda_\asq-\lambda_\sq}{20}}$ and $Z=\diag(1,-1)$.
The fiber losses of the channel are simulated by mixing the signal with vacuum at a beam splitter. We quantify the losses on Alice's and Bob's arm by $\eta_A$ and $\eta_B$ which specifies the reflectivity of the beam splitter, and thus, the amount of vacuum in the outgoing signal. We further include excess noise $\eta_\ex$ modeled as a classical Gaussian noise channel acting on the variances of quadratures as $V\mapsto V + \eta_\ex tV_{\vac}$ with $t$ the transmittance of the  channel and $V_{\vac}$ the variance of the vacuum (see, e.g.,~\cite{Lodewyck2007,Weedbrook2011}). This transforms the covariance matrix in \eqref{eq:CM} to
\begin{equation}\label{eq:LossModel}
\left( \begin{array}{cc}
\bar\eta_A \Gamma_A + (\eta_A+\eta_\ex\bar\eta_A )\Gamma_{\vac} & \sqrt{\bar\eta_A\bar\eta_B}\, \Gamma_{\text{cor}}  \\
 \sqrt{\bar \eta_A \bar\eta_B} \, \Gamma_{\text{cor}}  & \bar\eta_B  \Gamma_B+ (\eta_B+\eta_\ex \bar\eta_B)\Gamma_{\vac}
 \end{array} \right)
\end{equation}
where $\bar \eta_A = 1-\eta_A$, similar $\bar \eta_B$ and $\Gamma_\vac$ denotes the covariance matrix of the one-mode vacuum. 

In the protocol, the scaling factors for Alice's measurement $t_q$ and $t_p$ have to be adjusted. In an experiment, $t_p$ should be chosen such that the distance $d^{\PE}$ is small. A convenient way for that is to determine $\tilde Q_A$ and $\tilde P_A$ such that the second moments of Alice's and Bob's (continuous) amplitude and phase measurements match. These values can be determined locally and communicated in the classical post-processing step.  

The important parameter of the protocol that is directly related to the state is $d_0$, which should be chosen such that with high probability the distance $d^\PE$ computed for many samples of the Gaussian state given by the covariance matrix~\eqref{eq:LossModel} is smaller than $d_0$. 

The leakage in the reconciliation protocol $\ell_{\IR}$ is set to~\cite{Leverrier2010}
\begin{equation}
\ell_{IR} = H(X_B) - \beta I(X_B:X_A) \, ,
\end{equation}
where $H(X_B) $ denotes the Shannon entropy of $X_B$, $I(X_A:X_B)$ the mutual information between $X_A$ and $X_B$, and $\beta$ the efficiency of the reconciliation protocol. The efficiency in the Shannon limit is $\beta=1$, while $\beta<1$ for any finite $n$. 

It is now important that the protocol is robust, that is, it passes with high probability if no eavesdropper is presence. This means that the test $\cT(\alpha,T)$ has to pass with high probability for the above two-mode squeezed state. The probability that $\cT(\alpha,T)$ fails can be easily upper bounded by (see inequality~\eqref{eq:Tailbound})
\begin{equation}
\sqrt{ 8\pi} \sigma_t  N_\tot \text{e}^{-\alpha^2 /(2 \sigma_t^2)}  
\end{equation}
where $\sigma_t$ is the maximum of the standard deviations of the outcome distributions of $q_{t^1}$ and $ p_{t^2}$. Hence, by setting $\alpha = \sqrt{2\sigma_t^2 \ln(\sqrt{8\pi}\sigma_t N_\tot/\epsilon_\cT)}$ we ensure that the $\cT(\alpha,T)$ fails with probability smaller than $\epsilon_{\cT}$. Depending on $\alpha$ and $T$, we then choose $M$ such that $2\sqrt{2 n \Gamma(\alpha,T,M)}=\epsilon_2$ is smaller than $\epsilon_s$.

\begin{figure}\begin{center}\includegraphics*[width=8.8cm]{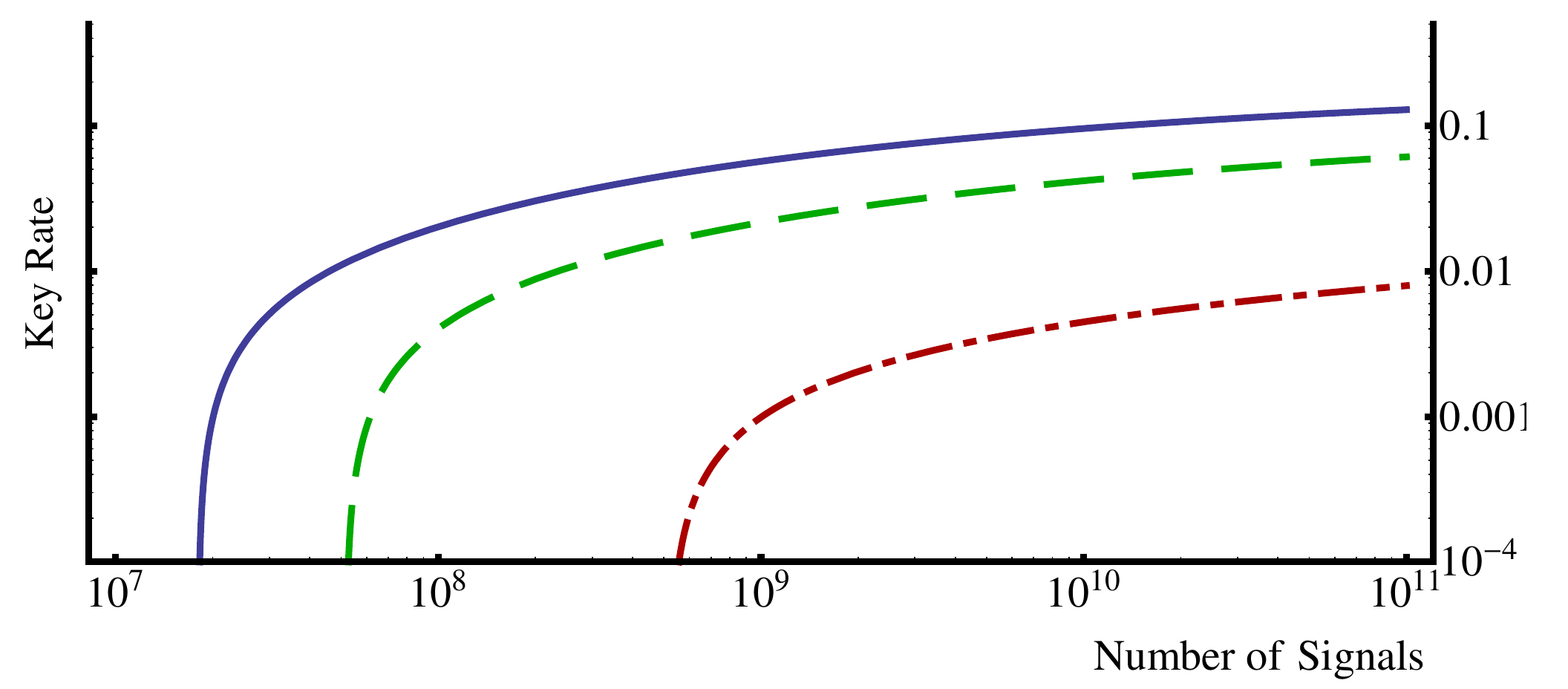}\caption{\label{fig:EC90} 
The plot shows the key rate $\ell/N_\tot$ for squeezing and antisqueezing of $11$dB and $16$dB and reconciliation efficiency $\beta =0.90$ depending on the number of signals $N_\tot$. Bob's total losses $\eta_B$ are $0.40$ (solid line), $0.45$ (dashed line) and $0.50$ (dash-dotted line). The other parameters are as in Figure~\ref{fig:EC95}. 
}\end{center}\end{figure}

We define the key rate as $\ell/N_\tot$ where $\ell$ is taken as in \eqref{thm,eq:KeyLengthCoherent} and optimized over the probability $r$ for choosing amplitude or phase. For that we simply express $n$, $k$, and $m$ in terms of $N_\tot$ and $r$. We further optimize the key rate over the spacing $\delta$ under the constraint $1\geq \delta \geq 0.01$ to account for the resolution of the detector. 
The security parameters are chosen as $\epsilon_s=\epsilon_c=10^{-9}$.
Moreover, we set $\epsilon_{\cT}=10^{-9}$, $T=0.99$ and $\epsilon_2 = \epsilon_s/10$ for which we find that $\alpha \leq 28 $ and $M\leq 8000$ in units of $\hbar =2$ for relevant values of $N_\tot$ and realistic squeezing strengths. 

In Figure~\ref{fig:EC95} and~\ref{fig:EC90} we plotted the key rate against the total number of exchanged signals $N_\tot$ for a reconciliation efficiency $\beta=0.95$ and $\beta=0.9$, respectively. The squeezing and antisqueezing is chosen as $\lambda_\sq=11$ and $\lambda_\asq = 16$ which has experimentally been achieved in the laboratory~\cite{Eberle11} at $1550$nm. Note that this squeezing values already include the efficiency of the homodyne detection. We further set the excess noise to $\eta_{\ex}=0.01$ in the plots.   
We note that a reconciliation efficiency of about $0.9$ is more realistic with current non-binary error correction codes.  
The maximal amount of losses to still obtain a secure key rate is slightly above $55$\% for $\beta =0.95$ and $50$\% for $\beta =0.9$. The key rate in dependence of the distance for different values of $\beta$ is plotted in Figure~\ref{fig:Dist}. For that we used a loss rate of $0.20$ dB per km and additional coupling losses of $0.05$. We see that for the same squeezing rates as above and an error correction efficiency of $0.95$, a positive key rate can be obtained for over $16$ km.

\begin{figure}\begin{center}\includegraphics*[width=8.8cm]{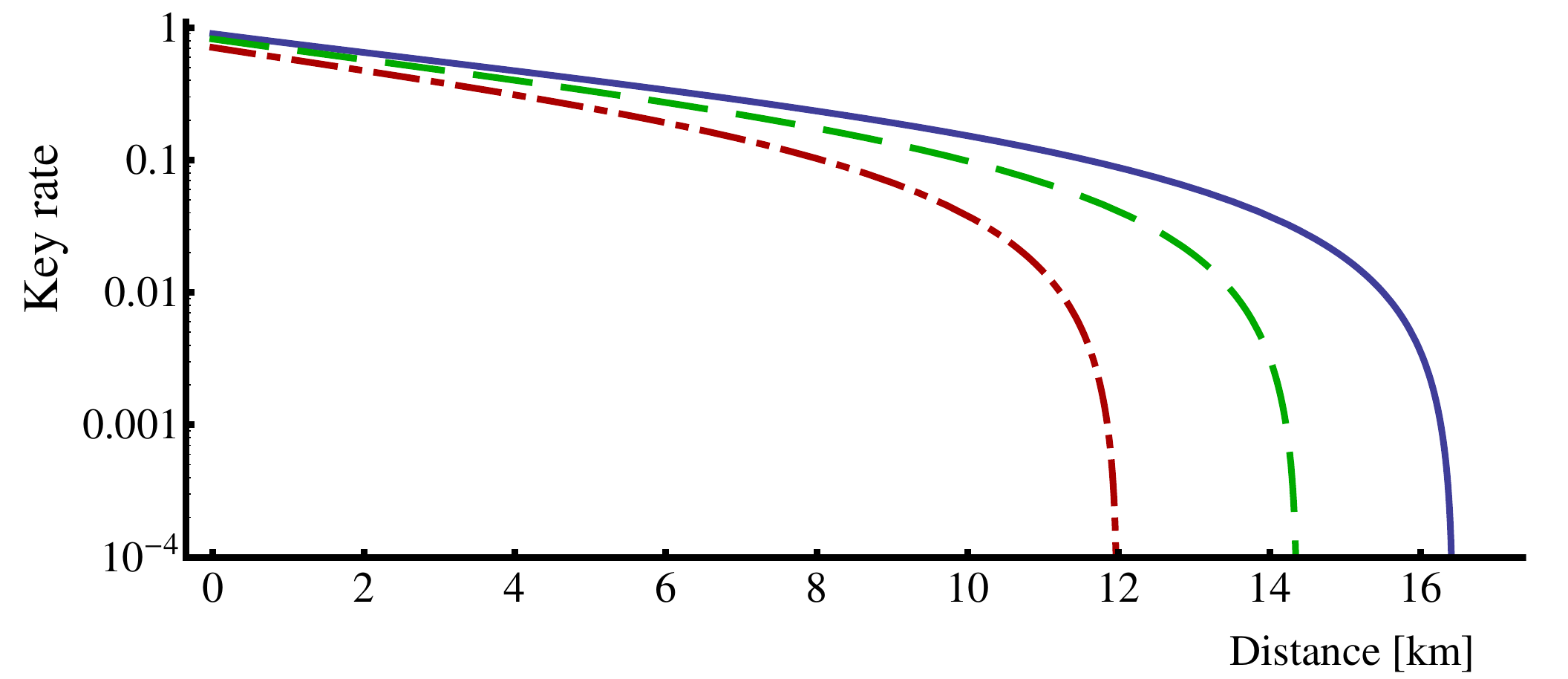}\caption{\label{fig:Dist} 
The key rate is plotted against the distance for $N_\tot=10^9$, squeezing and antisqueezing of $11$dB and $16$dB and reconciliation efficiency $\beta$ of $0.95$ (solid line), $0.90$ (dashed line) and $0.85$ (dash-dotted line). We assumed losses of $0.20$dB per km plus $0.05$ coupling losses. All the other parameters are as in Figure~\ref{fig:EC95}. 
}\end{center}\end{figure}

%%%%%%%%%%%%%%%%%%%%%%%%%%%%%%%%%%%%%%%%%%%%%%%%%%%%%
%%%%%%%%%%%%%%%%%%%%%%%%%%%%%%%%%%%%%%%%%%%%%%%%%%%%
%%%%%%%%%%%%%%%%%%%%%%%%%%%%%%%%%%%%%%%%%%%%%%%%%%%%%

\section{Security Analysis} \label{sec:SecAnalysis}

\subsection{Estimation of Eve's Information by the Uncertainty Principle with Quantum Memories}\label{sec:UR}
The first step of the security proof is the same as in~\cite{Furrer12} except that the roles of Alice and Bob are exchanged and that the basis choices for parameter estimation and key generation are different. We start with the definition of the min- and max-entropies. 

Let $X$ be a random variable over a countable set $\cX$ distributed according to $p_x$. Suppose further that $X$ is correlated to a quantum system B associated with Hilbert space $\cH_B$ and corresponding state space $\cS(\cH_B)=\{\rho_B |\, \rho_B\geq 0,\, \tr\rho_B = 1\}$.  
The min-entropy of a classical quantum state $\rho_{XB}=\sum_{x} p_x \kettbra x \otimes \rho_B^x$ with $\rho_B^x\in\cS(\cH_B)$ is defined as the negative logarithm of the optimal success probability to guess $X$ given access to the quantum memory $B$~\cite{koenig08}. In formulas, this is
\begin{equation}\label{minEnt}
H_{\min}(X|B)_\rho = -\log \Big( \sup_{\{E_x\}} \sum_x p_x \tr(E_x \rho_B^x) \Big) \, ,%\, | \, E_x\geq 0 \, , \ \sum_x E_x =\odtu 
\end{equation}
where the supremum is taken over all positive operator valued measures (POVM) $\{E_x\}$, i.e., $E_x\geq 0$ and $\sum_x E_x =\idty$. A further entropy related to the min-entropy via the uncertainty relation is the max-entropy which is defined as 
\begin{equation}
H_{\max}(X|B)= 2 \log \Big( \sup_{\sigma_B} \sum_x \sqrt{F(p_x\rho_B^x,\sigma_B)} \Big) \, ,
\end{equation} 
where the supremum runs over all states $\sigma_B\in\cS(\cH_B)$ and $F(\rho,\sigma)=(\tr \vert\sqrt{\rho}\sqrt{\sigma}\vert)^2 $ denotes the fidelity. 

The corresponding smooth min- and max-entropy are then obtained by optimizing the min- and max-entropy over nearby states. The closeness of states is measured with the purified distance $\cP(\rho,\sigma) = \sqrt{1-F(\rho,\sigma)}$~\cite{Tomamichel09}. We also allow for sub-normalized states defining the smooth min- and max-entropy as 
\begin{align}
H_{\min}^\epsilon(X|B)_\rho &= \sup_{\tilde\rho_{XB}} H_{\min}(X|B)_{\tilde\rho} \, , \\
H_{\max}^\epsilon(X|B)_\rho &= \inf_{\tilde\rho_{XB}} H_{\max}(X|B)_{\tilde\rho} \, ,
\end{align}
where the supremum and infimum are taken over sub-normalized states, i.e.,~$\tilde\rho_{XB} \geq 0$ and $\tr\tilde\rho_{XB} \leq 1$, with $\cP(\rho_{XB},\tilde\rho_{XB})\leq \epsilon$. 

Let us consider now the situation in the protocol. According to~\eqref{eq:KeyLength}, we have to bound the smooth min-entropy of the state associated with the raw key of Bob $X_B$ and the system of Eve $E$. Suppose that $\rho_{A^nB^nE}$ denotes the state of the $n$ modes on which the amplitude measurements for the raw key generation are performed conditioned on the event that the protocol passes. The state $\rho_{X_BE}$ of $X_B$ and $E$ can then be obtained by measuring the amplitudes of $B^n$ according to the discretization induced by the intervals $\{I_k\}$. But since the intervals $I_1$ and $I_{2M/\delta}$ are of infinite length any uncertainty relation will get trivial for the associated measurements. 

In order to avoid this problem, let us introduce phase and amplitude measurement with discretization $\{\tilde I_k\}_{k\in\mathbb Z}$, where 
\begin{align*}
\tilde I_k =(M +(k-1)\delta,-M +k\delta]  , \  k\in\mathbb Z \, .
\end{align*}
We note that $\tilde I_k = I_k$ for $k=2,3,...,2M/\delta-1$. We denote by $\tilde X_B$ ($\tilde Y_B^{\Key}$) the classical random variable corresponding to Bob's discretized amplitude (phase) measurement outcome $k\in\mathbb Z$. Moreover, the classical quantum state of $\tilde X_B$ ($\tilde Y_B^{\Key}$) and $A^nE$ is denoted by $\rho_{\tilde X_B A^nE}$ ($\rho_{\tilde Y_B^\Key A^nE}$). As we will see below, the energy test assures that the purified distance between $\rho_{ X_BE}$ and $\rho_{\tilde X_BE}$ as well as $\rho_{  Y^\Key_BA^n}$ and $\rho_{\tilde Y^\Key_BA^n}$ are small. 

Let us assume for now that $\cP(\rho_{ X_BE},\rho_{\tilde X_BE})$ and $\cP(\rho_{ Y^\Key_BA^n},\rho_{\tilde Y^\Key_BA^n})$ are smaller than $\tilde\epsilon$. We then find that 
\begin{align} \label{eq:boundMin}
H_{\min}^{\epsilon +\tilde\epsilon}(X_B|E)_\rho &\geq H_{\min}^{\epsilon}(\tilde X_B|E)_{\rho} \, , \\
H_{\max}^{\epsilon +\tilde\epsilon}(\tilde Y_B|A^n)_\rho &\leq H_{\max}^{\epsilon}(Y_B^\Key|A^n)_{\rho} \, ,\label{eq:boundMax}
\end{align} 
which is a simple consequence of the definition of smooth min- and max-entropy. The uncertainty relation in~\cite{Furrer11} then provides the inequality~\footnote{Note that the uncertainty relation in~\cite{Furrer11} was only proven for the non-smoothed min- and max-entropy. However, the extension of the inequality to smooth entropies is straightforward using similar arguments as in~\cite{tomamichel11}.}
\begin{equation} \label{eq:URsmooth}
H_{\min}^{\epsilon}(\tilde X_B|E)_{\rho} \geq - n \log c(\delta) - H_{\max}^{\epsilon}(\tilde Y_B|A^n)_{\rho} \, ,
\end{equation} 
with
\begin{equation}\label{eq:overlap2}
c(\delta) = \frac{1}{2\pi}\delta^2\cdot S_{0}^{(1)}\left(1,\frac{\delta^2}{4}\right)^{2} \, ,
\end{equation}
where $S_{0}^{(1)}(\cdot,x)$ is the $0$th radial prolate spheroidal wave function of the first kind. In the regime of interest $\delta \leq 1$, $c(\delta)$ can be approximated as in~\eqref{eq:overlap1}. 

If we combine now the inequalities~\eqref{eq:boundMin},~\eqref{eq:boundMax} and~\eqref{eq:URsmooth}, we obtain from the formula in~\eqref{eq:KeyLength} a lower bound on the key length given by
\begin{equation}\label{eq:KeyLength2}
-n\log c(\delta) - H_{\max}^\epsilon(Y_B^\Key|A^n)_\rho -  \ell_{\IR} -  \log \frac{1}{\epsilon_1^2\epsilon_c}  +2 \, ,
\end{equation}
 where $\epsilon \leq (\epsilon_1-\epsilon_s)/(2p_\pass) - 2\tilde\epsilon$. In the next section we use the energy test to give a bound on $\tilde\epsilon$.

%%%%%%%%%%%%%%%%%%%%%%%%%%%%%%%%%%%%%%%%%%%%%
%%%%%%%%%%%%%%%%%%%%%%%%%%%%%%%%%%%%%%%%%%%%%

\subsection{Failure Probabiltiy of the Energy Test}\label{sec:Test}
The goal of this section is to give a bound on the purified distance of $\rho_{ X_BE}$ and $\rho_{\tilde X_BE}$ as well as $\rho_{  Y^\Key_BA^n}$ and $\rho_{\tilde Y^\Key_BA^n}$. It turns out that they can be bounded by the probability that the energy test is passed although an amplitude or phase larger than $M$ is measured. We start with $\rho_{ X_BE}$ and $\rho_{\tilde X_BE}$. 

In a first step we compute that 
\begin{equation}\label{eq:Distance}
\cP(\rho_{ X_BE},\rho_{\tilde X_BE}) \leq \sqrt{1-\Pr[ \wedge_i \{|q_i| \leq M\} | \rho_{A^nB^nE}]^2} \, 
\end{equation}
where $\{|q_i| \leq M\} $ denotes the event that the absolute value of the continuous amplitude measurement of Bob's ith mode is smaller than $M$. This follows directly from the properties of the fidelity of a classical quantum state
\begin{align*}
{F(\rho_{ X_BE},\rho_{\tilde X_BE})}^{1/2} 
&=  \sum_{k=2}^{2M/\delta-1} F(p_k\rho^k_E, p_k\rho_E^k)^{1/2} \\
 &+  \sum_{k=1,{2M}/{\delta}} F(p_k\rho^k_E + q_k\sigma^k_E, p_k\rho_E^k)^{1/2} \\
 &\geq \sum_{k=1}^{2M/\delta-1} F(p_k\rho^k_E, p_k\rho_E^k)^{1/2}  \, \\
& = \sum_{k=1}^{2M/\delta-1} p_k \, ,
\end{align*}
where $p_k$ is the probability of measuring an amplitude in the interval $\tilde I_k$, $\rho_E^k$ the corresponding conditional state of Eve, and $q_i$, $\sigma^i_E$ for $i=1,2M/\delta$ similar for amplitude measurements smaller than $-M$ and larger than $M$, respectively. The inequality follows from $F(\rho +\sigma,\rho)^{1/2}\geq F(\rho,\rho)^{1/2}$ for any two non-normalized states $\rho$ and $\sigma$. Note now that the last line of the above computation is nothing else than $\Pr[ \wedge_i \{|q_i| \leq M\} | \rho_{A^nB^nE}]$ such that the bound~\eqref{eq:Distance} follows from the definition of the purified distance. 

We then denote the probability that Bob measures an amplitude larger than $M$ conditioned that the protocol passes by 
\begin{align}
p_\fail & = \Pr[  \neg \wedge_i \{|q_i| \leq M\} | \pass ] \\
 & = 1- \Pr[ \wedge_i \{|q_i| \leq M\} | \rho_{A^nB^nE}] \, , \label{eq: Prob1}
\end{align} 
where the second inequality follows since $\rho_{A^nB^nE}$ is the state conditioned that the protocol passes. Using $\neg \wedge_i \{|q_i| \leq M\} = \vee_i\{|q_i| > M\}$ and Bayes theorem, we obtain by simple manipulations
\begin{align*}
p_\fail & = \frac{1}{\ppass} \Pr [ \vee_i\{|q_i| > M\} \wedge \pass ]  \\
& \leq  \frac{1}{\ppass} \sum_i \Pr [|q_i |> M \wedge \pass ] \\
& \leq  \frac{1}{\ppass} \sum_i \Pr [|q_i| > M \wedge |q_{t^1_i}|\leq \alpha ] \, ,
\end{align*} 
where the last inequality holds since pass of the protocol implies that the energy test is passed which implies that $|q_{t^1_i}|\leq \alpha$. We now bound each term individually by 
\begin{align*}
&\Pr \big[|q_i| > M \wedge |q_{t^1_i}|\leq \alpha \big] \\
& = \int_{|x|\geq M} \Pr[q_i=x] \Pr \big[|q_{t^1_i}|\leq \alpha \ \big| \ q_i = x  \big] \ \dd x \\
& \leq \sup_{|x|\geq M}  \Pr \big[|q_{t^1_i}|\leq \alpha \ \big| \ q_i = x \big]  \, , 
\end{align*}
where the supremum in the last line refers to the essential supremum. 

We then show the following lemma. 
\begin{lem}\label{lem:FailureProb}
Let us assume that the energy test $\cT(\alpha,T)$ is passed and set $\mu=\sqrt{\frac{1-T}{2T}}$ and $\lambda = (\frac{2T-1}{T})^2$. If $\alpha\leq \mu M$, then it holds that $\sup_{|x|\geq M}  \Pr [|q_{t^1_i}|\leq \alpha \ | \ q_i = x ]$ is upper bounded by 
\begin{equation}\label{lem,eq:FailureProb}
 \Gamma(M,T,\alpha):=\frac{\sqrt{1+\lambda} + \sqrt{1+\lambda^{-1}}}{2} \exp\big({{\ -\frac{(\mu M - \alpha)^2}{T(1+\lambda)/2}}}\big)\, .
\end{equation}
\end{lem} 

\begin{proof}
In the following, we suppress the index $i$ since the argument applies independently to all possible incoming modes. We further label the different modes in the energy test setup as in Figure~\ref{fig:Energy}. 
We are interested in computing $\Lambda_x=\Pr [|q_{t^1_i}|\leq \alpha | \ q_i = x ]$ and without loss of generality we can assume that $x\geq 0$. 

In order to compute $\Lambda_x$, we write the characteristic function $\chi_{\text{out}}$ of the output state of modes $s'$, $t^1$, and $t^2$ in terms of the characteristic function $\chi_{\text{in}}$ of the input state of modes $a$, $b$, and $s$. Let $B$ be the matrix describing the linear transformation of the coordinates of the phase space induced by the beam splitters, that is, $r_{\text{out}} = B r_{\text{in}}$, where $r_{\text{in}}=(q_a,p_a,q_b,p_b,q_s,p_s)$ and $r_{\text{out}}=(q_{s'},p_{s'},q_{t^1},p_{t^1},q_{t^2},p_{t^2})$. For the following it will be important that $q_{s'} = \sqrt{T}q_s + \sqrt{1-T} q_a$ and $q_{t^1}= \sqrt{1/2} q_b + \sqrt{T/2} q_a + \sqrt{(1-T)/2} q_s$.

We then have that $\chi_{\text{out}}(r_{\text{out}}) = \chi_{\text{in}}(B^{-1}r_{\text{out}})$, where $\chi_{in}(r_{\text{in}}) = \chi_{\vac}(q_a,p_a)  \chi_{\vac}(q_b,p_b) \chi_{s}(q_s,p_s) $ has product form. Integrating over all output modes under the condition $|q_{s'}|\leq \alpha$ and changing variables $r_{\text{in}} = B r_{\text{out}}$, we obtain that the probability $\Pr [|q_{t^1_i}|\leq \alpha ]$ is given by 
\begin{align}
\int_{\tilde A} \chi_{\vac}(q_a) \chi_{\vac}(q_b)  \chi_s(q_s) \ \dd q_a \ \dd q_b \ \dd q_s  \, ,
\end{align}
 where $\chi_{*}(q)=\int \ \dd p \chi_{*}(q,p)$ and $\tilde A$ is determined by the condition
\begin{equation}
|q_{t^1}|=| \sqrt{1/2} q_b + \sqrt{T/2} q_a + \sqrt{(1-T)/2} q_s | \leq \alpha \, .
\end{equation} 
In order to condition on $q_{s'} = x$, we set $\chi_s(q_s) = \delta( q_s - [\sqrt{1/T} x + \sqrt{(1-T)/T} q_a])$ where $\delta$ denotes the Dirac delta distribution and we used that $q_s' = \sqrt{T}q_s + \sqrt{1-T} q_a$. 
Hence, integrating over $q_s$ results in 
\begin{align}\label{eq:int1}
\Lambda_x  \leq \int_{A}\chi_{\vac}(q_a) \chi_{\vac}(q_b) \ \dd q_a \ \dd q_b  \, ,
\end{align}
where $A=\{(q_a,q_b) | \ d_1q_a + d_2 q_b + \mu x \leq \alpha \} $. Here, we obtained $A$ from $\tilde A$ by setting $q_s=\sqrt{1/T} x + \sqrt{(1-T)/T}q_a$ and removing the absolute value. 

In order to bound the integral in~\eqref{eq:int1}, we split the area $A$ into $A_1=A\cap\{q_a \geq 0\} $ and $A_2 = A\backslash A_1$. If we set $l(q_b)=\max \{0, 1/d_1(\mu x-\alpha -d_2q_b)\}$, we get that the integration over $A_1$ amounts to
\begin{align}
& \frac{1}{2\pi} \int_{-\infty}^\infty \ \dd q_b \ e^{-q_b^2/2 }\int_{l(q_b)}^\infty \dd q_a \ e^{-q_a^2/2} \\ 
&\leq \frac{1}{2\sqrt{2\pi}} \int_{-\infty}^\infty \ \dd q_b \ e^{-q_b^2/2-l(q_2)^2/2} \, ,\label{eq:int2}
\end{align}
where the inequality follows from 
\begin{equation}\label{eq:Tailbound}
\int_l^\infty e^{-q^2/2} \dd q \ \leq \sqrt{\pi/2} \ e^{-l^2/2}
\end{equation}
for $l\geq 0$. A straightforward calculation of~\eqref{eq:int2} gives
\begin{align}\label{eq:Bound1}
 \frac12 \sqrt{1+\lambda^{-1}} \exp{\Big( -\frac{(\mu x - \alpha)^2}{T(1+\lambda)/2}\Big)} \, . 
\end{align}

In order to compute the integral over $A_2$, we note first that $A_2=\{(q_a,q_b) | q_a\leq 0 \ , \ -\infty < q_b \leq u(q_a) \}$ with $u(q_a)=1/d_2[d_1 q_a -(\mu x-\alpha)]$. Using that $u(q_a)\leq 0$ for all $q_a\leq 0$, we can apply again~\eqref{eq:Tailbound} to bound 
\begin{align}
& \frac{1}{2\pi} \int_{-\infty}^0\ \dd q_a \ e^{-q_a^2/2 }\int_{-\infty}^{u(q_a)}\dd q_b \ e^{-q_b^2/2} \\ 
&\leq \frac{1}{2\sqrt{2\pi}} \int_{-\infty}^\infty \ \dd q_a \ e^{-q_a^2/2-u(q_a)^2/2} \, , \label{eq:int3}
\end{align}
where we also extended the integration over $q_a$ to run over the whole real line. Finally, the same calculation as before shows that~\eqref{eq:int3} is given by 
\begin{align}\label{eq:Bound2}
 \frac12 \sqrt{1+\lambda} \exp{\Big( -\frac{(\mu x - \alpha)^2}{T(1+\lambda)/2}\Big)} \, .
\end{align}
We can thus conclude that $\Lambda_x$ is bounded by the sum of~\eqref{eq:Bound1} and~\eqref{eq:Bound2}. Finally, the supremum over $x$ is attained for $x=M$ which completes the proof. 
\end{proof}

By means of Lemma~\ref{lem:FailureProb}, we can now bound $p_\fail\leq n\Gamma(M,T,\alpha)/p_\pass$. Using the relation in~\eqref{eq:Distance} together with $(1-p_\fail)^2 \geq 1-2 p_\fail$, we finally arrive at
\begin{align}
\cP(\rho_{ X_BE},\rho_{\tilde X_BE})   \leq \sqrt{\frac{ { 2 n \ \Gamma(M,T,\alpha)}}{{p_\pass}}} \, .
\end{align}

Let us now consider the case of $\rho_{ Y^\Key_BA^n}$ and $\rho_{\tilde Y^\Key_BA^n}$. It is easy to see that the same strategy can be applied as in the previous situation. This is simply based on the fact that $|p_{t^2}| \leq \alpha$ if the test $\cT(\alpha,M)$ is passed. Hence, following the exactly same steps for the phase measurements as before for amplitude, we find that also 
\begin{align}
\cP(\rho_{  Y^\Key_BA^n},\rho_{\tilde Y^\Key_BA^n})   \leq \sqrt{\frac{ { 2 n \ \Gamma(M,T,\alpha)}}{{p_\pass}}} \, ,
\end{align}
holds. 

Summarizing the above arguments, we have thus shown that~\eqref{eq:KeyLength2} is a lower bound on the key rate if we set 
\begin{equation}\label{eq:tildeEps}
\tilde \epsilon = \sqrt{\frac{ { 2 n \ \Gamma(M,T,\alpha)}}{{p_\pass}}} \, .  
\end{equation}

%%%%%%%%%%%%%%%%%%%%%%%%%%%%%%%%%%%%%%%%%%%%%%%%%%%%%
%%%%%%%%%%%%%%%%%%%%%%%%%%%%%%%%%%%%%%%%%%%%%%%%%%%%

\subsection{Statistical Estimation of the Max-Entropy}  

The goal of this section is to use the information from the parameter estimation step to upper bound the smooth max-entropy $H_{\max}^\epsilon(Y_B^\Key|A^n)_\omega$. In a first step, we apply Alice's scaled and discretized phase measurement to $A^n$ mapping it to a classical outcome $Y^\Key_A$ also in $\cX^n$. Using now that the smooth max-entropy can only increase under processing of the side-information~\cite{Tomamichel09,Furrer11}, we obtain that 
\begin{equation}\label{eq:DataPr}
H^\epsilon_{\max}(Y_B^\Key|A^n)_\rho \leq H^\epsilon_{\max}(Y_B^\Key|Y^\Key_A)_\rho \, .
\end{equation}

We next note that it has been shown in~\cite{Furrer12} that if $X$ and $Y$ are random variables on $\cX^n\times \cX^n$ distributed according to $Q_{XY}$ for which $\text{Pr}_{Q}[d(X,Y)\geq d] \leq \epsilon^2$ holds, it follows that 
\begin{equation}\label{HmaxBound}
H^\epsilon_{\max}(X|Y)_Q \leq n \log \gamma(d) \, ,
\end{equation}
with $\gamma$ as defined in~\eqref{eq:gamma}. In order to apply this result to bound $H^\epsilon_{\max}(Y_B^\Key|Y^\Key_A)_\rho$, we have to find an estimation of $d^\Key=d(Y_B^\Key,Y^\Key_A)$ that holds with probability $\epsilon^2$. For that we use a large deviation bound and estimate the probability that $d^\Key=d(Y_B^\Key,Y^\Key_A)$ is larger than $d_0 +\mu$ where conditioned on pass $d_0\geq d^{\PE}=d(X_A^{\PE},X_B^{\PE})$. But since the alphabet size scales with $M$ and is thus very large, a direct application of a large deviation bound would result in a large failure probability. This can be avoided by employing a strategy that splits the problem into two estimation steps.

In the first step, we bound in Lemma~\ref{lem:LargeDeviation1} the probability that $\text{m}_2(Y_A^\Key)$ is larger than $V_{Y_A}^\PE+\nu$, respectively, that $\text{m}_2(Y_B^\Key)$ is larger than $V_{Y_B}^\PE+\nu$. This will be done using Serfling's large deviation bound~\cite{Serfling74}. Given that $\text{m}_2(Y_A^\Key) \leq V_{Y_A}^\PE+\nu$ and $\text{m}_2(Y_B^\Key) \leq V_{Y_B}^\PE+\nu$, we can bound the average variance of the distance $d(Y_B^\Key,Y_A^\Key)$ on $Y_B^\Key \times Y_A^\Key$, and thus, of the total population $Y_A^\tot \times Y_B^\tot $ formed by $Y_B^\Key\times Y_A^\Key $ and $Y_B^\PE \times Y_A^\PE$. Indeed, denoting $N=n+k$, we can bound the average variance of the population by 
\begin{align*}
\sigma^2 & = \frac 1 N\sum_i \vert (Y_A^\tot)_i - (Y_B^\tot)_i \vert^2  - d(Y_B^\tot,Y_A^\tot)^2 \\
& \leq \frac k N V^\PE_d  + \frac 1N \sum_i \vert (Y_A^\Key)_i - (Y_B^\Key)_i \vert^2 - (\frac k N d^\PE)^2  \\
& \leq \frac k N (V^\PE_d - \frac k N (d^\PE)^2) + \frac 1N \sum_i (\vert (Y_A^\Key)_i\vert + \vert( Y_B^\Key)_i\vert )^2 
\end{align*}
where we used that $d^\tot = \frac k N d^\PE + \frac n N d^\Key$. Applying the Cauchy-Schwarz inequality, we can then bound $\sum_i (\vert (Y_A^\Key)_i\vert + \vert( Y_B^\Key)_i\vert )^2$ by 
\begin{align}
k \Big(\text{m}_2(Y_A^\Key) +  \text{m}_2(Y_B^\Key) + \big(\text{m}_2(Y_A^\Key) \text{m}_2(Y_B^\Key)\big)^{\frac 12} \Big) \, .
\end{align}
Hence, given that $\text{m}_2(Y_A^\Key) \leq V_{Y_A}^\PE+\nu/\delta^2$ and $\text{m}_2(Y_B^\Key) \leq V_{Y_B}^\PE+\nu/\delta^2$ holds, we find that $\sigma \leq \sigma_*$ with $\sigma_*$ as defined in~\eqref{eq:Sigma}. 

In the second step, we bound in Lemma~\ref{lem:LargeDeviation2} the probability that $d^\Key=d(Y_B^\Key,Y^\Key_A)$ is larger than $d^\PE +\mu$ for a fixed and bounded $\sigma$. Combining these two steps, we can then estimate
\begin{align} \nonumber
 \pr{d^\Key \geq d_0 + \mu | \pass }    
 & \leq \pr{d^\Key \geq d^\PE + \mu | \pass }  \\ \nonumber
  & \leq  \frac 1{p_\pass}  \pr{d^\Key \geq d^\PE + \mu } 
\\  & \leq \frac 1{p_\pass} \Big( \pr{\text{m}_2(Y_A^\Key) > V_{Y_A}^\PE+\nu}\  \nonumber
 \\ \nonumber
& \quad +   \pr{\text{m}_2(Y_B^\Key) > V_{Y_B}^\PE+\nu}  \\  & \quad +  \pr{d^\Key \geq d^\PE + \mu | \text{C}}\Big) \label{eq:Prob1}
\end{align}
where $C$ denotes the condition $\text{m}_2(Y_B^\Key) \leq V_{Y_B}^\PE+\nu$ and $\text{m}_2(Y_A^\Key) \leq V^\PE_{\tilde P_A}$.

\begin{lem}\label{lem:LargeDeviation1} 
Let  $Y$ be a string in $\cX^{n+m}$ and $Y^{P}$ be a random sample without replacement from $Y$ of length $m$ with $\text{m}_2(Y^P)=V_{Y}^\PE$. Then, for the  average second moment of the remaining sample $Y^\Key$ of length $n$, holds that
\begin{align}
\pr{\text{m}_2( Y^\Key) \geq V_{Y}^\PE + \nu } 
 \leq \exp\Big(\frac{-2 \nu^2\delta^4 n m^2}{M^4(n+m) (m+1)} \Big) \, .
\end{align}
\end{lem}
\begin{proof}
The proof is similar to strategies applied in~\cite{tomamichellim11,Furrer12} and based on a large deviation bound for random sampling without replacement by Serfling~\cite{Serfling74}.
Denoting the population mean of the variance by $V_{Y}=\text{m}_2(Y)$ and $V^\Key_{Y}=\text{m}_2(Y^\Key)$, we have that 
\begin{equation}\label{eq:Population}
n V^\Key_{Y} + m V_{Y}^\PE = (n+m) V_{Y} \, .
\end{equation}
The large deviation bound in~\cite{Serfling74} implies that $\pr{V^\Key_{Y} \geq V_{Y} +\tilde \nu}$ is upper bounded by  
\begin{equation}
\exp\big( -\frac{2\tilde\nu^2 n(n+m)}{(M/\delta)^4(m+1)}\big) \, .
\end{equation}
Since the bound is independent of $V_{Y}$, it is not necessary to know the actual value of $V_{Y}$. Indeed, using the relation in~\eqref{eq:Population}, we obtain the desired bound
\begin{align*}
\pr{ V_{Y}^\Key) \geq V_{Y}^\PE + \nu } &\leq \pr{V^\Key_{Y} \geq V_{Y} + \frac{m}{m+n} \nu} \\
& \leq \exp\big(\frac{-2\nu^2 n m^2}{(M/\delta)^4(n+m) (m+1)} \big)\, .
\end{align*}
\end{proof}

\begin{lem}\label{lem:LargeDeviation2} 
Let $Y_A^\tot \times Y_B^\tot $ be in $(\cX\times \cX)^N$ with $d_\tot= d(Y_A^{\tot},Y_B^{\tot})$ and $Y_A^\PE \times Y_B^\PE$ a random sample from it without replacement of length $k$ with $d^{\PE}=d(Y_A^{\PE},Y_B^{\PE})$. Let further $\sigma^2= \sum_i |(Y_A^\tot)_i-(Y_B^\tot)_i|^2 - d_\tot^2 $ be the average variance of the population.  
Then, for $d^\Key=d(Y_A^{\Key},Y_B^{\Key})$ of the remaining sample $Y_A^\Key\times Y_B^\Key$ of length $n=N-k$, holds that
\begin{align}\label{lem,eq:Bernstein}
\pr{d^\Key \geq d^\PE + \nu } 
 \leq \exp\Big(\frac{-\mu^2 n (k/N)^2  }{2\sigma^2 + 4\mu/3(k/N)(M/\delta)} \Big) \, .
\end{align}
\end{lem} 
\begin{proof}
The bound follows directly from Bernstein's inequality
\begin{equation}
\pr{d^\Key \geq d^\tot + \tilde\mu } \leq  \exp\big(- \frac{n \tilde\mu^2}{2\sigma^2 + 2\mu\vert \cX\vert/3} \big) \, ,
\end{equation}
which, as shown by Hoeffding~\cite{Hoeffding1963}, also holds for sampling without replacement. Using that $nd^\Key + kd^\PE= Nd^\tot$ and that $|\cX|=2M/\delta$, a straigthforward calculation results in~\eqref{lem,eq:Bernstein}. 
\end{proof}

We are now ready to prove Theorem~\ref{thm:KeyLength}. For that we observe that from~\eqref{HmaxBound} follows that 
\begin{equation}
H_{\max}^\epsilon(Y_B^\Key|Y^\Key_A)_\rho \leq \gamma(d_0+\mu) \, ,
\end{equation}
if $\mu$ is such that~\eqref{eq:Prob1} is smaller than $\epsilon^2$. Hence, we use Lemma~\ref{lem:LargeDeviation1} and Lemma~\ref{lem:LargeDeviation2} to bound~\eqref{eq:Prob1} and set the expression equal to $\epsilon^2$, where $\epsilon\leq (\epsilon_1-\epsilon_s)/(2p_\pass) - 2 \tilde \epsilon$ with $\tilde \epsilon$ as in \eqref{eq:tildeEps} (c.f.~\eqref{eq:KeyLength2}). Solving the equation for $\mu$ and using $p_\pass \leq 1$, we obtain an upper bound on $\mu$ by~\eqref{eq:mu}. This concludes the security proof. 

%%%%%%%%%%%%%%%%%%%%%%%%%%%%%%%%%%%%%%%%%%%%%%%%%%%%%%%%%%%%%
%%%%%%%%%%%%%%%%%%%%%%%%%%%%%%%%%%%%%%%%%%%%%%%%%%%%%%%%%% 

\section{Performance and Limitations of Security Proofs based on the Extended Uncertainty Principle} \label{sec:Tightness}

In Section~\ref{sec:UR}, we have seen that the main ingredient in the security proof is the uncertainty relation with quantum memory for smooth min- and max-entropy (c.f.~\eqref{eq:URsmooth})
\begin{equation} \label{eq:URsmoothEnt}
H_{\min}^{\epsilon}(Q_B^{\delta,n}|E)_{\rho} +  H_{\max}^{\epsilon}(P_B^{\delta,n}|A^n)_{\rho}  \geq  - n \log c(\delta) \, .
\end{equation} 
Here, we denote by $Q_B^{\delta,n}$ and $P_B^{\delta,n}$ the classical random variable induced by an arbitrary amplitude and phase measurement with discretization into intervals of equal length $\delta$.
Thus, the tightness of the bound on the optimal key rate~\eqref{eq:KeyLength} crucially depends on how tight the uncertainty relation is for the state given in the protocol. Since we are interested in optimality in the following, and as such in the question of how much key can be extracted under normal working condition, we can assume that Eve is absent for the moment. Then, the state is in good approximation given by the $n$-fold tensor product of identical Gaussian states described by a covariance matrix depending on coupling and channel losses as well as excess noise as described in~\eqref{eq:LossModel}.

But even though we can assume that the state takes this simple form it is still very hard to compute the corresponding smooth min- and max-entropy directly. We circumvent this problem by using a further approximation. In particular, we can use the asymptotic equipartition property in infinite dimensions~\cite{Furrer10}, saying that the smooth min-entropy $\frac 1n H_{\min}^{\epsilon}(Q_B^{\delta,n}|E)_{\rho^{\otimes n}}$ can be approximated up to a correction $\cO (\frac1{\sqrt{n}})$ by the von Neumann entropy $ H(Q_B^{\delta}|E)_{\rho}$. Here,  $\rho_{Q_B^{\delta} E}$ is given by measuring the amplitude with a spacing $\delta$ on a single copy. The same applies for the smooth max-entropy such that $\frac 1n H_{\max}^{\epsilon}(P_B^{\delta,n}|A)_{\rho^{\otimes n}}$ can be approximated by $ H(P_B^{\delta}|A)_{\rho}$. 

Furthermore, if we choose $\delta$ small enough we can approximate the von Neumann entropy of the discrete distribution over intervals of length $\delta$ by the differential von Neumann entropy~\cite{Berta13} 
\begin{equation}\label{eq:DeltaApprox}
H(Q_B^{\delta}|E)_{\rho} \approx h(Q_B|E) -\log \delta \, ,
\end{equation}
where $h(Q_B|E)$ denotes the differential quantum conditional entropy of the continuous amplitude measurement. 
Similarly, we have that $ H(P_B^{\delta}|E)_{\rho}  \approx h(P_B|A) -\log \delta $. Using that $c(\delta)\approx \delta^2/(2\pi)$, we can thus conclude that in the asymptotic limit inequality~\eqref{eq:URsmoothEnt} is well approximated by 
\begin{equation} \label{eq:URvN}
h(Q_B|E) + h(P_B|A) \geq  \log 2\pi \, . 
\end{equation}

Hence, we can qualitatively investigate the tightness of~\eqref{eq:URsmoothEnt} by considering inequality~\eqref{eq:URvN}. For our situation, the latter one can now easily be analyzed as the differential quantum conditional entropy can be computed for Gaussian classical and quantum states. In the following, we always choose system $E$ as the Gaussian purification of the Gaussian state between $A$ and $B$. In~\cite{Berta13}, it was shown that~\eqref{eq:URvN} gets approximately tight for a two-mode squeezed state without losses and squeezing above $10$ dB. However, tightness holds only conditioned on Alice's quantum system but not after she performs the amplitude measurements. The data processing inequality only ensures that 
\begin{equation}
h(P_B|A) \leq h(P_B|P_A) \, , 
\end{equation}
but equality does not always hold, even for the optimal measurement on $A$. Unfortunately, in our case it turns out that the loss through the data processing inequality is substantial (see Figure~\ref{fig:UR}) such that the optimality of the bound has to be analyzed for the inequality after applying the data processing inequality 
\begin{equation} \label{eq:URvNdp}
h(Q_B|E) + h(P_B|P_A) \geq  \log 2\pi \, . 
\end{equation}
In Figure~\ref{fig:UR}, we plotted the tightness of~\eqref{eq:URvN} and~\eqref{eq:URvNdp} for the same parameters of the state for which the key rates are plotted in Section~\ref{sec:KeyRates}. We see that unfortunately, the gap between the left hand side and right hand side of~\eqref{eq:URvN} and~\eqref{eq:URvNdp} increases for high losses. We further note that an increase of the initial squeezing does hardly change the gap for losses above $30$\%.

\begin{figure}\begin{center}\includegraphics*[width=8.8cm]{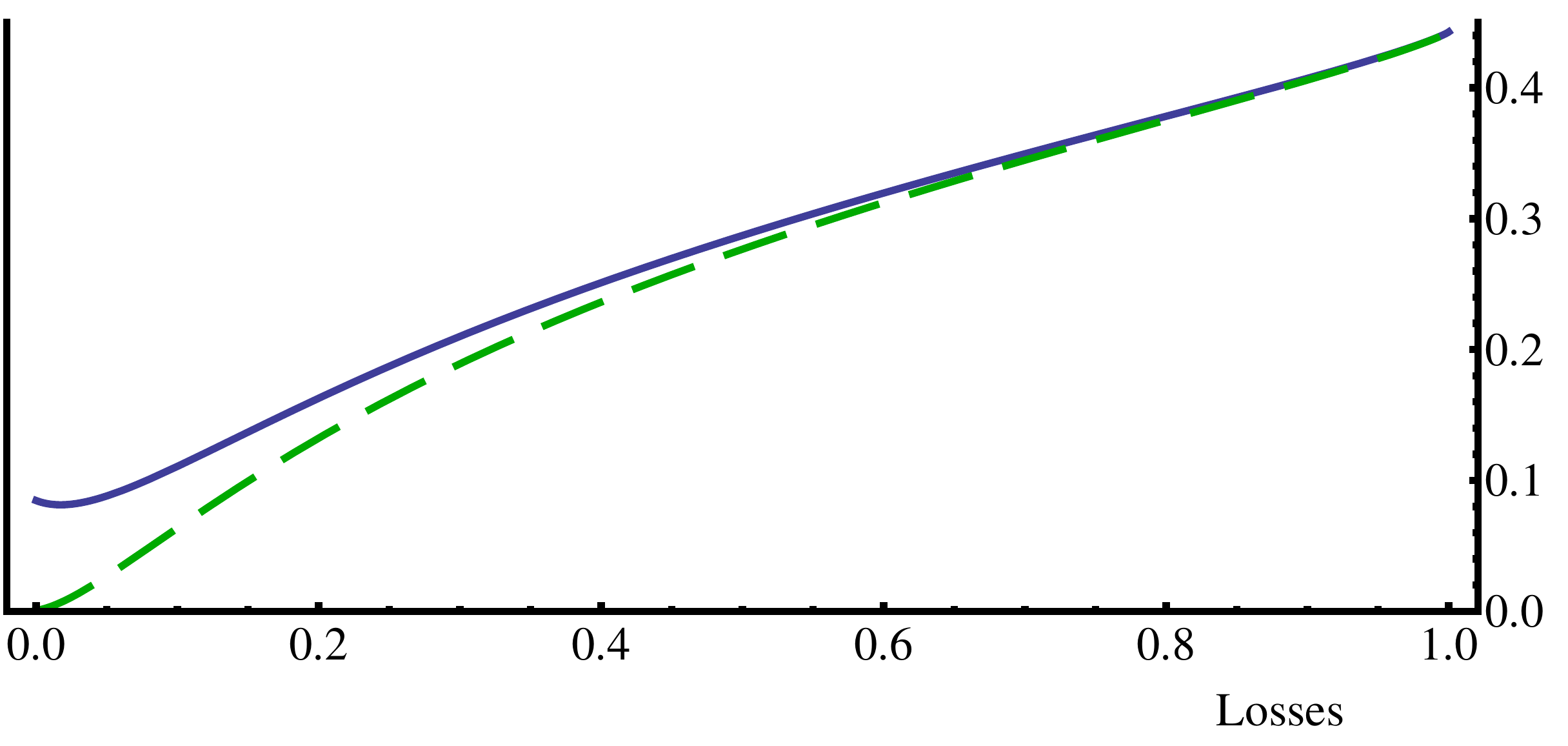}\caption{\label{fig:UR} 
The gap between l.h.s. and r.h.s. of the uncertainty relations in~\eqref{eq:URvNdp} (solid line) and~\eqref{eq:URvN} (dashed line) are plotted for a two-mode squeezed state with squeezing and antisqueezing of $11$dB and $16$dB against the losses on Bob's mode. The losses on Alice's mode and the excess noise are set to $0$. The gap for~\eqref{eq:URvNdp} (solid line) is the amount by which the bound on the key rate reduces in the asymptotic limit compared with the optimal key rate. 
}\end{center}\end{figure}

This gap severely limits the tolerated noises also causing that the finite-key rates presented in Section~\ref{sec:KeyRates} vanish for high losses. We can quantitatively analyze the effect of the untightness of the uncertainty relation on the key rate by calculating the asymptotic key rate~\eqref{eq:KeyLength}. In this regime all the statistical estimation errors disappear and collective attacks are as strong as coherent attacks~\cite{Renner_Cirac_09}. We find for the asymptotic key rate by using that $\ell_\IR = H(P_B^{\delta}|P_A^{\delta})$ for perfect error correction and~\eqref{eq:DeltaApprox} the simple formula 
\begin{equation}\label{eq:AsymKeyUR}
r_{\text{UR}} =  \log  2\pi - 2 h(P_B|P_A) \, .
\end{equation}
In contrast, the asymptotically optimal key rate given by the Devetak-Winter formula~\cite{DevetakW} is 
\begin{equation}\label{eq:AsymKey}
r_{\text{Opt}} = h(P_B|E) - h(P_B|P_A) \, , 
\end{equation}
where we also applied the approximation in~\eqref{eq:DeltaApprox}. In Figure~\ref{fig:Asym}, we compare the finite-key rate from~\eqref{thm,eq:KeyLengthCoherent} with the asymptotic key rates $r_{\text{UR}}$ and $r_{\text{Opt}}$ for the same parameters as in Figure~\ref{fig:EC95} except that the excess noise is set equal to $0$. We see that even the asymptotic key rate $r_{\text{UR}}$ vanishes for moderate losses of $66$\%. We remark that even if the squeezing is arbitrarily high and the losses in Alice's mode are $0$\% the maximally tolerated losses are not exceeding $75$\%. 

\begin{figure}
\begin{center}\includegraphics*[width=8.8cm]{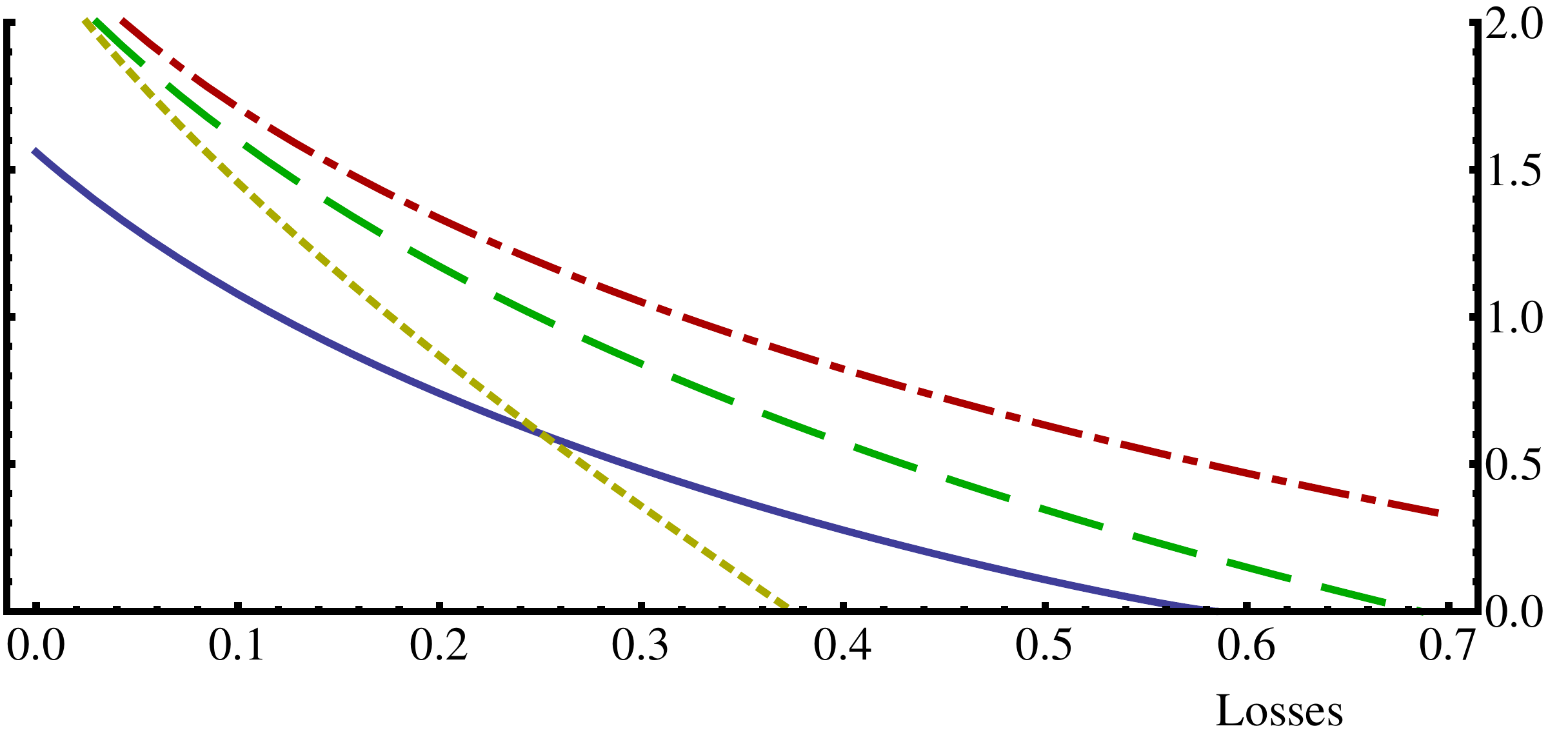}\caption{\label{fig:Asym} 
The loss dependence of the finite-key rate for $N_\tot=10^{11}$ (straight line) is compared with the asymptotic key rate $r_{\text{UR}}$ (dashed), the optimal key rate $r_{\text{Opt}}$ (dashed-dotted), and the asymptotic key rate for direct reconciliation $r_{\text{DR}}$ (dotted). The squeezing and antisqueezing is set to $11$dB and $16$dB and Alice's coupling losses as well as the excess noise to $0$. The reconciliation efficiency of the non-asymptotic key rate is $\beta=0.95$ and the other parameters are as in Figure~\ref{fig:EC95}. 
}\end{center}\end{figure}

In Figure~\ref{fig:Asym}, we also plotted the asymptotic key rate obtained via the extended uncertainty principle if using a direct reconciliation protocol. In this situation the asymptotic key rate is $r_{\text{DR}} =  \log  2\pi - 2 h(P_A|P_B)$. The plot shows that we have obtained a finite-key rate in the case of reverse reconciliation for losses much larger than what is ultimately tolerated in the case of direct reconciliation.

\section{Conclusion}\label{sec:Conclusion}

We have presented a security proof against coherent attacks including finite-size effects for a reverse reconciliation continuous variable QKD protocol. The protocol is based on the generation of two-mode squeezed states and homodyne detection. Security for transmission losses of up to $50$\% for experimental parameters demonstrated in~\cite{Eberle11} have been certified under realistic assumptions. A remaining challenging point in an implementation of the presented protocol will be the reconciliation protocol. However, recently some advances have been made in non-binary error correction codes such that reconciliation efficiencies above $90$\% seem realistic.   

We further investigated on the tightness of the security analysis based on the uncertainty relation with quantum memory and showed that even in the asymptotic limit the maximally tolerated losses are bounded. The reason for that is that the uncertainty relation is not perfectly tight and for high losses the trade-off between Eve's knowledge and the correlations between Bob and Alice gets very small. Hence, for the high loss regime a very tight bound on Eve's information is crucial.

 \emph{Acknowledgements.---} I gratefully acknowledge valuable discussions with Joerg Duhme, Vitus H\"andchen and Takanori Sugiyama. I'm especially grateful to Anthony Leverrier who provided very helpful comments on a first version and proposed using a test like in Figure~\ref{fig:Energy} to control the energy of Eve's attack. This work is supported by the Japan Society for the Promotion of Science (JSPS) by KAKENHI grant No. 24-02793.

%\bibliography{libraryQKD}
%%\bibliographystyle{h-physrev}
%\bibliographystyle{unsrt}

\end{document}